\documentclass[letterpaper, 10 pt,journal]{ieeetran}

\IEEEoverridecommandlockouts                               

\usepackage{graphics} 
\usepackage{graphicx}
\usepackage{epsfig} 
\usepackage{stfloats}
\usepackage{amsmath} 
\usepackage{amssymb}  

\usepackage{tikz} 
\usepackage{pgfplots} 

\usepackage{mleftright}
\usepackage{algpseudocode}
\usepackage{algorithm}
\usepackage{epstopdf}
\usepackage{arydshln}

\usepackage{color}

\usepackage{lipsum}
\usepackage{verbatim}

\usepackage{amsthm}

\usepackage{mathtools}

\newtheorem{theorem}{\bf Theorem} \newtheorem{definition}[theorem]{\bf Definition} 
 \newtheorem{remark}[theorem]{\bf Remark}
 \newtheorem{corollary}[theorem]{\bf Corollary}  
\newtheorem{assumption}[theorem]{\bf Assumption}

\title{Robust data-driven state-feedback design}

\author{Julian Berberich$^{1}$, Anne Koch$^{1}$, Carsten W. Scherer$^{2}$, and Frank Allg\"ower$^{1}$
\thanks{
This work was funded by Deutsche Forschungsgemeinschaft (DFG, German Research Foundation) under Germany's Excellence Strategy - EXC 2075 - 390740016.
The authors thank the International Max Planck Research School for Intelligent Systems (IMPRS-IS) for supporting Julian Berberich and Anne Koch.}
\thanks{$^{1}$J. Berberich, A. Koch and F. Allg\"ower are with the Institute for Systems Theory and Automatic Control, University of Stuttgart, 70569 Stuttgart, Germany. {\tt\small E-mail: \{julian.berberich, anne.koch, frank.allgower\}@ist.uni-stuttgart.de}}
\thanks{$^{2}$C. W. Scherer is with the Institute of Mathematical Methods in the Engineering Sciences, Numerical Analysis and Geometrical Modeling, Department of Mathematics, University of Stuttgart, 70569 Stuttgart, Germany. {\tt\small E-mail: carsten.scherer@mathematik.uni-stuttgart.de}}
}

\begin{document}
\IEEEpubid{\begin{minipage}{\textwidth}\ \\[12pt] This version has been accepted for publication in Proc. American Control Conference (ACC), 2020. Personal use of this material is permitted. Permission from AACC must be obtained for all other uses, in any current or future media, including reprinting/republishing this material for advertising or promotional
purposes, creating new collective works, for resale or redistribution to servers or lists, or reuse of any copyrighted component of this work in other works.
\end{minipage}}

\maketitle

\begin{abstract}
We consider the problem of designing robust state-feedback 
controllers for discrete-time linear time-invariant systems, based directly on measured data.
The proposed design procedures require no model knowledge, but only a single open-loop data trajectory, which may be affected by noise.
First, a data-driven characterization of the uncertain class of closed-loop matrices under state-feedback is derived.
By considering this parametrization in the robust control framework, we design data-driven state-feedback gains with guarantees on stability and performance, containing, e.g., the $\mathcal{H}_\infty$-control problem as a special case.
Further, we show how the proposed framework can be extended to take partial model knowledge into account.
The validity of the proposed approach is illustrated via a numerical example.
\end{abstract}


\section{Introduction}
\label{sec:introduction}
Recently, the design of controllers directly from measured data has received increasing interest~\cite{Hou13,Recht18}.
While established methods, e.g., those based on reinforcement learning, rarely address closed-loop guarantees, there has been a renewed effort to provide such guarantees using novel statistical estimation techniques~\cite{Matni19b,Matni19,Boczar18,Dean19}.
Potential alternatives are, e.g., robust control with prior set membership identification~\cite{Milanese91}, which is however well-known to be computationally demanding, and unfalsification-based approaches~\cite{Kosut01}, which typically require infinitely long data for closed-loop guarantees.
In general, providing non-conservative end-to-end guarantees for the closed loop using \emph{noisy data of finite length} is an open problem, even if the data is generated by a linear time-invariant (LTI) system.
\IEEEpubidadjcol

A promising approach towards this goal relies on behavioral systems theory.
In~\cite{Willems05}, it was proven that the vector space of all input-output trajectories of an LTI system is spanned by time-shifts of a single measured trajectory, given that the respective input signal is persistently exciting.
Thus, a single data trajectory can be used to characterize an LTI system, without any prior identification steps.
Recently, there have been various contributions which consider this result in the context of data-driven system analysis and control, including 
dissipativity verification from measured data~\cite{Romer19} or an extension of~\cite{Willems05} to certain classes of nonlinear systems~\cite{Berberich20a}.
Moreover, the recent work~\cite{Persis19} derives a simple data-dependent closed-loop parametrization of LTI systems under state-feedback.
This parametrization is used to solve various control problems from data, including stabilization and linear-quadratic regulation.
However, no meaningful guarantees were given in the presence of noisy data.

It is the goal of this paper to provide non-conservative end-to-end guarantees for data-driven control.
To be more precise, we employ a single noisy input-state trajectory of finite length to design controllers which guarantee closed-loop stability and performance for all systems which are consistent with the measured data and the assumed noise bound.
This is achieved by extending the approach of~\cite{Persis19} to account for noise and applying robust control techniques to the resulting uncertain system class.
Another recent paper~\cite{Waarde19} considers data-driven analysis and control with not persistently exciting data.
In particular, it is shown for noise-free data that certain control problems can be solved from data, even if the system cannot be uniquely identified, thus illustrating advantages of direct data-driven control.
Similarly, the results of this paper do not require persistence of excitation explicitly.
Moreover, our results lead to simple design procedures for direct data-driven control with desirable closed-loop guarantees, and are thus a promising alternative to identification-based control.

The paper is structured as follows.
After stating the problem formulation in Section~\ref{sec:preliminaries}, we use noisy data to describe the uncertain closed loop under state-feedback, and we apply known robust control methods to design controllers with stability and performance guarantees in Section~\ref{sec:state_feedback}.
Moreover, we extend the proposed, purely data-driven approach to systems with mixed data-driven and model-based components.
In Section~\ref{sec:example}, we apply the robust state-feedback design techniques successfully to an unstable example system.
The paper is concluded in Section~\ref{sec:conclusion}.

\section{Preliminaries}
\label{sec:preliminaries}
We denote the $n\times n$ identity matrix by $I_n$, where the index is omitted if the dimension is clear from the context.
Further, $A^\perp$ denotes a matrix containing a basis of the kernel of $A$.
We write $\ell_2$ for the space of square-summable sequences.
In a linear matrix inequality (LMI), $*$ represents blocks, which can be inferred from symmetry.

\newpage

Moreover, we define, for elements $\{x_k\}_{k=i}^{i+L+N-2}$ of a sequence $x$, the Hankel matrix
\begin{align*}
X_{i,L}^N\coloneqq\begin{bmatrix}x_i & x_{i+1}& \dots & x_{i+N-1}\\
x_{i+1} & x_{i+2} & \dots & x_{i+N}\\
\vdots & \vdots & \ddots & \vdots\\
x_{i+L-1} & x_{i+L} & \dots & x_{i+L+N-2}
\end{bmatrix}.
\end{align*}
That is, the matrix $X_{i,L}^N$ starts with the element $x_i$ and has $L$ rows and $N$ columns. 
As a shorthand notation, we abbreviate $N$-windows of $x$, starting at $i=0$ and $i=1$, by
\begin{align*}
X&= X_{0,1}^N=\begin{bmatrix}x_0 & x_1& \dots & x_{N-1}\end{bmatrix},\\
X_+&=X_{1,1}^N=\begin{bmatrix}x_1 & x_{2}& \dots & x_{N}\end{bmatrix},
\end{align*}
respectively.
%
%
In the present paper, we consider LTI systems of the form
%
\begin{align}\label{eq:sys}
\mleft[
\begin{array}{c}x_{k+1}\\\hline
z_k
\end{array}
\mright]=
\mleft[\begin{array}{c|cc}A_{tr}&B_w&B_{tr}\\\hline
C&D_{w}&D
\end{array}
\mright]
\mleft[
\begin{array}{c}x_{k}\\\hline
w_k\\ u_k
\end{array}
\mright],
\end{align}
where $x_k\in\mathbb{R}^n$ is the state, $w_k\in\mathbb{R}^{m_w}$ is the disturbance, $u_k\in\mathbb{R}^m$ is the control input, and $z_k\in\mathbb{R}^{p_z}$ is the performance output. 
We design state-feedback controllers $u_k=Kx_k$ to control the system~\eqref{eq:sys}.
Our design procedures are purely data-driven and do not require knowledge of the \emph{true} system matrices $A_{tr},B_{tr}$.
We do, however, assume that the matrices $B_w,C,D_w,D$ are known.
For our purposes, $B_w$ is essentially a parameter to model the influence of the disturbance, whereas $C,D_{w},D$ constitute a user choice for performance.
In this paper, we use the following definition of persistence of excitation of the data under consideration.
\begin{definition}\label{def:pe}
The sequence $\{x_k,u_k\}_{k=0}^{N-1}$ is called persistently exciting if the matrix $\begin{bmatrix}X\\U\end{bmatrix}$ has full row rank.
\end{definition}
%
According to~\cite{Willems05}, controllability and a certain rank property of the input are sufficient for persistence of excitation.
\begin{theorem}[{\cite[Corollary 2]{Willems05}}]\label{thm:Willems}
If $\left(A_{tr},\begin{bmatrix}B_{tr}&B_w\end{bmatrix}\right)$ is controllable and the matrix
\begin{align*}
\begin{bmatrix}W_{0,n+1}^{N-n}\\U_{0,n+1}^{N-n}\end{bmatrix}
\end{align*}
has full row rank, then $\{x_k,u_k\}_{k=0}^{N-1}$ is persistently exciting.
\end{theorem}
Definition~\ref{def:pe} differs from the notion of persistence of excitation considered in~\cite{Willems05}, which concerns only the input data, and is preferred in the present paper for convenience.
In~\cite{Persis19}, it is shown how a single, persistently exciting open-loop trajectory can be employed to recover the system matrices of an LTI system.
Furthermore, a linear parametrization of the closed loop under state-feedback is derived, depending also only on a single open-loop data trajectory.
It is the contribution of the present paper to extend the framework of~\cite{Persis19} in order to provide robust stability and performance guarantees in the presence of noise.
In contrast to~\cite{Persis19}, persistence of excitation will generally not be required for our results.

Throughout this paper, we consider the following scenario:
From simulation or an experiment, a single open-loop input-state sequence $\{x_k,u_k\}_{k=0}^{N}$ is obtained as a trajectory of~\eqref{eq:sys} for some \emph{unknown} disturbance $\{\hat{w}_k\}_{k=0}^{N-1}$.
This trajectory is used directly for robust controller design, without prior system identification.
The only available information on the disturbance realization is the following bound on the matrix 
\begin{align*}
\hat{W}=\begin{bmatrix}\hat{w}_0&\hat{w}_1&\dots&\hat{w}_{N-1}\end{bmatrix}.
\end{align*}
%
\begin{assumption}\label{ass:disturbance_bound}
The matrix $\hat{W}$ is an element of
\begin{align*}
\mathcal{W}=\Big\{W\in\mathbb{R}^{m_w\times N}\Big\lvert
\begin{bmatrix}W\\I\end{bmatrix}^\top \begin{bmatrix}
Q_w&S_w\\S_w^\top&R_w\end{bmatrix}
\begin{bmatrix}W\\I\end{bmatrix}\succeq0\Big\},
\end{align*}
for some known matrices $Q_w\in\mathbb{R}^{m_w\times m_w}$, $S_w\in\mathbb{R}^{m_w\times N}$, $R_w\in\mathbb{R}^{N\times N}$ with $R_w\succ0$.
\end{assumption}
%
Through Assumption~\ref{ass:disturbance_bound} it is assumed that the unknown disturbance realization, which affects the measured data, lies in some known set which is described by a quadratic matrix inequality.
Implicitly, $\hat{W}\in\mathcal{W}$ implies a quadratic bound on the sequence $\{\hat{w}_k\}_{k=0}^{N-1}$ and encompasses many practical bounds as special cases.
For instance, if the maximal singular value of $\hat{W}$ is bounded as $\sigma_{\max}(\hat{W})\leq\bar{w}$, then $\hat{W}\in\mathcal{W}$ holds with $Q_w=-I$, $S_w=0$, $R_w=\bar{w}^2 I$.
More generally, a description of the form $\hat{W}\in\mathcal{W}$ provides a flexible framework to model general noise signals, in particular when multiple quadratic matrix inequalities are combined.
It is an interesting aspect for future research to derive suitable matrices $Q_w,S_w,R_w$ for different, practically relevant scenarios such as norm bounds on the sequence $\{\hat{w}_k\}_{k=0}^{N-1}$.


\section{Data-driven state-feedback}\label{sec:state_feedback}
In this section, we consider the design of state-feedback gains, based directly on measured data which is perturbed by a disturbance satisfying Assumption~\ref{ass:disturbance_bound}.
First, we derive a data-driven characterization of the uncertain closed loop, using a single open-loop data trajectory.
Thereafter, we apply known robust control methods to this parametrization in order to design state-feedback controllers which guarantee stability and performance for all closed-loop matrices that are consistent with the measured data.
Finally, we extend the proposed framework to systems with mixed data-driven and model-based components.

\subsection{Uncertain closed-loop parametrization}\label{sec:state_feedback_uncertain}
In the following, we extend~\cite{Persis19} by characterizing the closed-loop dynamics of~\eqref{eq:sys} under state-feedback, using noisy measurements.
Let $\{x_k,u_k\}_{k=0}^{N}$ be a measured trajectory of~\eqref{eq:sys}, corresponding to an unknown disturbance realization $\hat{W}$.
We define $\Sigma_{X,U}$ as the set of all pairs $(A,B)$ that are consistent with the data $\{x_k,u_k\}_{k=0}^{N}$ for \emph{some} noise instance $W\in\mathcal{W}$, i.e.,
\begin{align*}
\Sigma_{X,U}=\{(A,B)\mid X_+=AX+BU+B_wW,\>W\in\mathcal{W}\}.
\end{align*}
Using \emph{fixed} data matrices $X$ and $U$, $\Sigma_{X,U}$ parametrizes the unknown system matrices $A$ and $B$ via $\mathcal{W}$.
By assumption, the true disturbance realization $\hat{W}$ satisfies $X_+=A_{tr}X+B_{tr}U+B_w\hat{W}$ and $\hat{W}\in\mathcal{W}$; therefore, the true pair $(A_{tr},B_{tr})$ is an element of $\Sigma_{X,U}$.
Furthermore, for some state-feedback gain $K$, we define the set of closed-loop matrices that are consistent with the data as
\begin{align*}
\Sigma_{X,U}^K=\{A_K\mid A_K=A+BK,(A,B)\in\Sigma_{X,U}\}.
\end{align*}
In the following, we show that an exact parametrization of $\Sigma_{X,U}^K$ can be constructed directly from open-loop data.
To this end, for some matrix $G\in\mathbb{R}^{N\times n}$, we define $\mathcal{A}_G$ as the set of matrices $A_G\in\mathbb{R}^{n\times n}$ such that
\begin{align}\label{eq:thm_uncertain_matrix1}
A_G=(X_+-B_wW)G,
\end{align}
for some $W\in\mathcal{W}$ satisfying
\begin{align}\label{eq:thm_uncertain_matrix2}
(X_+-B_wW)\begin{bmatrix}X\\U\end{bmatrix}^\perp=0.
\end{align}
\begin{theorem}\label{thm:uncertain}
If $G\in\mathbb{R}^{N\times n}$ and $K\in\mathbb{R}^{m\times n}$ satisfy
\begin{align}\label{eq:thm_uncertain_equality}
\begin{bmatrix}X\\U\end{bmatrix}G=\begin{bmatrix}I\\K\end{bmatrix},
\end{align}
then $\Sigma_{X,U}^K=\mathcal{A}_G$. 
\end{theorem}
\begin{proof}
First, we note that the constraint~\eqref{eq:thm_uncertain_matrix2} is equivalent to the implication
\begin{align*}
\begin{bmatrix}X\\U\end{bmatrix}\tilde{V}=0\quad\Rightarrow\quad (X_+-B_wW)\tilde{V}=0,
\end{align*}
for any matrix $\tilde{V}$ with $N$ rows.
By the Fredholm alternative, this is in turn equivalent to the existence of a solution $V$ to the system of linear equations
%
%
%
\begin{align}\label{eq:thm_uncertain_proof1}
V\begin{bmatrix}X\\U\end{bmatrix}=X_+-B_wW.
\end{align}
\textbf{Proof of $\mathbf{\Sigma_{X,U}^K\subseteq\mathcal{A}_G}$:}
Let $A_K\in\Sigma_{X,U}^K$, i.e., there exist matrices $A,B$ as well as $W\in\mathcal{W}$ such that
\begin{align}\label{eq:thm_uncertain_proof3}
A_K&=A+BK,\\\label{eq:thm_uncertain_proof4}
X_+&=AX+BU+B_wW.
\end{align}
Then, it follows that
\begin{align*}
A_K\stackrel{\eqref{eq:thm_uncertain_proof3}}{=}A+BK&=\begin{bmatrix}A&B\end{bmatrix}\begin{bmatrix}I\\K\end{bmatrix}
\stackrel{\eqref{eq:thm_uncertain_equality}}{=}\begin{bmatrix}A&B\end{bmatrix}\begin{bmatrix}X\\U\end{bmatrix}G\\
&\stackrel{\eqref{eq:thm_uncertain_proof4}}{=}(X_+-B_wW)G.
\end{align*}
It remains to show that $W$ satisfies~\eqref{eq:thm_uncertain_matrix2} or, equivalently, there exists $V$ such that~\eqref{eq:thm_uncertain_proof1} holds.
It follows directly from~\eqref{eq:thm_uncertain_proof4} that $V=\begin{bmatrix}A&B\end{bmatrix}$ solves~\eqref{eq:thm_uncertain_proof1}, which thus proves $A_K\in\mathcal{A}_G$.\\
\textbf{Proof of $\mathbf{\mathcal{A}_G\subseteq\Sigma_{X,U}^K}$:}
Let $A_G\in\mathcal{A}_G$, i.e., there exists $W\in\mathcal{W}$ such that~\eqref{eq:thm_uncertain_matrix1} and~\eqref{eq:thm_uncertain_matrix2} hold.
We need to show the existence of matrices $A,B$ as well as $\tilde{W}\in\mathcal{W}$ such that
\begin{align*}
A+BK&=(X_+-B_wW)G,\\
X_+&=AX+BU+B_w\tilde{W}.
\end{align*}
If we choose $\tilde{W}=W$, these equations are equivalent to
\begin{align*}
\begin{bmatrix}A&B\end{bmatrix}
\begin{bmatrix}X&I\\U&K\end{bmatrix}=(X_+-B_wW)\begin{bmatrix}
I&G\end{bmatrix}.
\end{align*}
Using~\eqref{eq:thm_uncertain_equality}, this is in turn equivalent to
\begin{align}\label{eq:thm_uncertain_proof5}
\begin{bmatrix}A&B\end{bmatrix}
\begin{bmatrix}X\\U\end{bmatrix}\begin{bmatrix}I&G\end{bmatrix}=(X_+-B_wW)\begin{bmatrix}
I&G\end{bmatrix}.
\end{align}
Since $A_G\in\mathcal{A}_G$, there exists a solution $V$ to~\eqref{eq:thm_uncertain_proof1}.
Hence, the choice $\begin{bmatrix}A&B\end{bmatrix}=V$ satisfies~\eqref{eq:thm_uncertain_proof5}, which implies $A_G\in\Sigma_{X,U}^K$.
\end{proof}

Theorem~\ref{thm:uncertain} provides an exact parametrization of the uncertain closed loop under a fixed state-feedback $K$, using a single \emph{open-loop} trajectory of the unknown system.
In particular, no closed-loop measurements and no model knowledge are required to construct the set $\mathcal{A}_G$, which parametrizes the uncertain closed loop.
This set relies on fixed data matrices $X$ and $U$, which are obtained offline, and is parametrized via the disturbance $W\in\mathcal{W}$ satisfying~\eqref{eq:thm_uncertain_matrix2}.
The equation~\eqref{eq:thm_uncertain_matrix2} ensures that the matrices in $\mathcal{A}_G$ contain only those $W\in\mathcal{W}$ for which there exist matrices $A,B$ satisfying the system dynamics.

In general, the condition~\eqref{eq:thm_uncertain_equality} only requires that $X$ has full row rank, but not necessarily that the data are persistently exciting.
Nevertheless, if $\{x_k,u_k\}_{k=0}^{N-1}$ is persistently exciting, then, for any state-feedback $K$,~\eqref{eq:thm_uncertain_equality} can be solved for $G$, i.e., any possible closed-loop matrix can be constructed.
Equivalently, the set of all $A_G\in\mathcal{A}_G$ with $G\in\mathbb{R}^{N\times n}$ satisfying $XG=I$ is equal to the set of all possible closed-loop matrices under state-feedback.
\begin{corollary}\label{cor:sets}
If $\{x_k,u_k\}_{k=0}^{N-1}$ is persistently exciting, then it holds that
\begin{align}\label{eq:cor_sets}
\begin{split}
\{A_G\in\mathcal{A}_G\mid G\in &\>\mathbb{R}^{N\times n},XG=I\}\\
=&\{A_K\in\Sigma_{X,U}^K\mid K\in\mathbb{R}^{m\times n}\}.
\end{split}
\end{align}
\end{corollary}
\begin{proof}
This follows directly from Theorem~\ref{thm:uncertain}.
\end{proof}
%
%
%
Corollary~\ref{cor:sets} suggests that the set $\mathcal{A}_G$ can be employed to design controllers with robustness guarantees for all closed-loop matrices in $\Sigma_{X,U}^K$, by optimizing over the parameter $G$ instead of the gain $K$.
If the data are not persistently exciting, then~\eqref{eq:cor_sets} holds with "$\subseteq$" instead of "$=$", since $\mathcal{A}_G$ contains only closed-loop matrices resulting from feedback gains $K$ of the form $K=UG$ (compare~\eqref{eq:thm_uncertain_equality}).
In this case, Theorem~\ref{thm:uncertain} can still be used for robust controller design since, for any \emph{fixed} $K=UG$, $\mathcal{A}_G$ captures the full closed-loop uncertainty induced by the noise, i.e., $\mathcal{A}_G=\Sigma_{X,U}^K$.
However, the conservatism of robust controller design based on Theorem~\ref{thm:uncertain} increases if the data are not persistently exciting since there may exist a controller $K$ which, e.g., renders all matrices in $\Sigma_{X,U}^K$ stable, but for which there exists no $G$ satisfying~\eqref{eq:thm_uncertain_equality}.

The disturbance $W$ parametrizing $\mathcal{A}_G$ is not only restricted by $W\in\mathcal{W}$ but also via the affine constraint~\eqref{eq:thm_uncertain_matrix2} and therefore, the construction of $\mathcal{A}_G$ requires the computation of the kernel of $\begin{bmatrix}X^\top&U^\top\end{bmatrix}^\top$, which may be undesirable from a numerical viewpoint.
In Sections~\ref{sec:state_feedback_stability} and~\ref{sec:state_feedback_performance}, we employ a superset of $\mathcal{A}_G$ to derive simple robust controller design procedures for closed-loop stability and performance, respectively.



\subsection{Robust state-feedback for stability}\label{sec:state_feedback_stability}

In this section, we apply known robust control methods to render all matrices in $\mathcal{A}_G$ stable.
To facilitate the design, we consider
\begin{align}\label{eq:uncertain_set_relaxed}
\mathcal{A}_G^s=\{A_G\mid A_G=(X_+-B_wW)G,\>W\in\mathcal{W}\},
\end{align}
which is a superset of the uncertain closed loop $\mathcal{A}_G$, i.e., $\mathcal{A}_G\subseteq\mathcal{A}_G^s$.
The difference between $\mathcal{A}_G^s$ and $\mathcal{A}_G$ is that the latter considers only those disturbances $W\in\mathcal{W}$, which satisfy the $n$ constraints defined by~\eqref{eq:thm_uncertain_matrix2}.
Hence, $\mathcal{A}_G^s$ is in general larger than $\mathcal{A}_G$ and, therefore, controller design based on $\mathcal{A}_G^s$ is generally more conservative than a design based on $\mathcal{A}_G$.
Nevertheless, $\mathcal{A}_G^s$ admits a simpler parametrization and can be translated directly into a standard robust control format.
Further, as we will see in Section~\ref{sec:example}, considering $\mathcal{A}_G^s$ instead of $\mathcal{A}_G$ leads to meaningful robust controllers also for practical examples.
Providing an exact quantification of the conservatism induced by this replacement will be the subject of future research.

In the following, we exploit that the parametrization $\mathcal{A}_G^s$ is equivalent to a particular lower linear fractional transformation (LFT) (compare~\cite[Chapter 10]{Zhou96}).
%
%
To be more precise, 
the matrices in $\mathcal{A}_G^s$ can be described as a lower LFT of a nominal closed-loop system depending on $G$ with the disturbance $W$, i.e.,
\begin{align}\label{eq:LFT_W}
\begin{split}
\begin{bmatrix}x_{k+1}\\\tilde{z}_k\end{bmatrix}&=
\begin{bmatrix}X_+G&B_w\\-G&0\end{bmatrix}
\begin{bmatrix}x_k\\\tilde{w}_k\end{bmatrix},\\
\tilde{w}_k&=W\tilde{z}_k,
\end{split}
\end{align}
%
where $W\in\mathcal{W}$.
It follows from Theorem~\ref{thm:uncertain} that, if $G$ satisfies $XG=I$, the above LFT contains all potential closed-loop systems under control with state-feedback $K=UG$.
The following result exploits this fact by using robust control methods to design a stabilizing controller parameter $G$ for the LFT~\eqref{eq:LFT_W}, which hence stabilizes all elements of $\Sigma_{X,U}^K$.

\begin{corollary}\label{cor:stab}
If there exist $\mathcal{X}\succ0,G\in\mathbb{R}^{N\times n}$ such that
\begin{align}\label{eq:cor_stab_equality}
XG&=I
\end{align}
as well as
\begin{align}\label{eq:cor_stab_LMI}
\mleft[
\begin{array}{cc}*&*\\ *&*\\\hline *&*\\ *&*
\end{array}
\mright]^\top
\mleft[
\begin{array}{cc|cc}-\mathcal{X}&0&0&0\\
0&\mathcal{X}&0&0\\\hline
0&0&Q_w&S_w\\0&0&S_w^\top&R_w
\end{array}
\mright]
&\mleft[
\begin{array}{cc}
I&0\\X_+G&B_w\\\hline0&I\\-G&0
\end{array}
\mright]\prec0,
\end{align}
then $A+BK$ with $K=UG$ is stable for all $(A,B)\in\Sigma_{X,U}$.
\end{corollary}
\begin{proof}
This follows from an application of known robust control methods to the system~\eqref{eq:LFT_W} (cf.~\cite{Scherer00,Scherer00b}).
\end{proof}

Corollary~\ref{cor:stab} applies known robust control methods to design state-feedback controllers which robustly stabilize all elements of $\mathcal{A}_G^s$.
If $K$ is designed according to Corollary~\ref{cor:stab}, then~\eqref{eq:thm_uncertain_equality} holds and hence, Theorem~\ref{thm:uncertain} leads to $\Sigma_{X,U}^K=\mathcal{A}_G$ which thus implies $\Sigma_{X,U}^K=\mathcal{A}_G\subseteq\mathcal{A}_G^s$.
This guarantees stability of all closed-loop matrices $\Sigma_{X,U}^K$ that are consistent with the measured data.
Similar to Theorem~\ref{thm:uncertain}, Corollary~\ref{cor:stab} does not require persistently exciting data explicitly.
Thus, it may be possible to find a controller $K$ which stabilizes all elements of $\Sigma_{X,U}^K$, even if persistence of excitation does not hold, i.e., if the data is not sufficiently rich for system identification.
Similar phenomena were analyzed for system analysis and control from noise-free data in~\cite{Waarde19}, where also full row rank of $X$ was sufficient to design stabilizing controllers from data.

Nevertheless, persistence of excitation is required for equality in~\eqref{eq:cor_sets}, i.e., to construct any closed-loop system (cf. Corollary~\ref{cor:sets}), and thus, it enhances feasibility of~\eqref{eq:cor_stab_LMI}.
In particular, if the data are persistently exciting and there exists a controller which stabilizes all matrices in $\mathcal{A}_G^s$ with a common Lyapunov function, then~\eqref{eq:cor_stab_equality} and~\eqref{eq:cor_stab_LMI} are feasible.
Hence, Corollary~\ref{cor:stab} contains two main sources of conservatism:
a) the difference between $\mathcal{A}_G^s$ and $\Sigma_{X,U}^K$ and b) the fact that a \emph{common} Lyapunov function is employed for stabilization, similar to simple model-based robust controller design methods.
Nevertheless, Corollary~\ref{cor:stab} provides computationally tractable conditions, based directly on open-loop data, to design controllers with stability guarantees.

\begin{remark}\label{rk:LMI}
Although~\eqref{eq:cor_stab_LMI} is not an LMI, it is routine to transform it into one following the same steps as in model-based robust state-feedback design (compare~\cite{Scherer00,Scherer00b}).
To be more precise, after performing a congruence transformation on~\eqref{eq:cor_stab_LMI} with $\text{diag}(\mathcal{X}^{-1},I)$ and applying the Schur complement twice, the nonlinear matrix inequality~\eqref{eq:cor_stab_LMI} leads to the LMI
\begin{align}\label{eq:cor_stab_LMI_Schur}
\begin{bmatrix}-\mathcal{Y}&-M^\top S_w^\top&M^\top X_+^\top &M^\top\\
-S_wM&Q_w&B_w^\top& 0\\
X_+M&B_w&-\mathcal{Y}&0\\
M&0&0&-R_w^{-1}\end{bmatrix}\prec0
\end{align}
in the variables $\mathcal{Y}=\mathcal{X}^{-1},M=G\mathcal{X}^{-1}$.
Further, multiplying~\eqref{eq:cor_stab_equality} by $\mathcal{Y}$ from the right yields the linear equality constraint $XM=\mathcal{Y}$.
Together with the LMI~\eqref{eq:cor_stab_LMI_Schur}, this leads to a semidefinite program whose feasibility can be checked using standard solvers.
The stabilizing state-feedback gain can then be recovered as $K=UM\mathcal{Y}^{-1}$.
\end{remark}

Corollary~\ref{cor:stab} suggests a valuable alternative to sequential system identification and stabilizing robust control.
In particular, in the presence of deterministic noise, identification-based methods are usually either computationally intractable, overly conservative, or they admit no guarantees from finite data.
Essentially, Corollary~\ref{cor:stab} is a computationally tractable alternative to robust controller design based on set membership estimation, which relies on an explicit construction of the set $\Sigma_{X,U}$~\cite{Milanese91}.
Further alternatives include unfalsification-based control, which typically requires infinitely long data and a prescribed controller structure for closed-loop guarantees~\cite{Kosut01,Safonov97}, or a stochastic setting, where recent work has addressed finite-time guarantees on system identification with sequential robust control~\cite{Matni19b,Matni19,Boczar18,Dean19}.
The latter results are based on sophisticated statistical analysis and many of them rely on restrictive assumptions, such as the availability of multiple independent data trajectories, each of which only supplies one data tuple to the estimator.
On the contrary, our approach relies on simple matrix manipulations combined with existing robust control methods and requires only a single data trajectory of finite length.
Despite these advantages, the presented approach requires state measurements which may be restrictive in practice.
Extending the results in this paper to input-output data is an important aspect for future research.

\begin{remark}
For the state-feedback stabilization problem under additive state measurement noise,~\cite{Persis19} provides sufficient conditions for closed-loop stability.
However, this result relies on assumptions that cannot be verified from measured data.
Moreover, in contrast to the approach of~\cite{Persis19}, an extension of Corollary~\ref{cor:stab} to more general (robust) control objectives is straightforward.
\end{remark}

\subsection{Robust state-feedback for performance}\label{sec:state_feedback_performance}

Next, we consider the system~\eqref{eq:sys} including the performance channel $w\mapsto z$.
The goal is to use data $\{x_k,u_k\}_{k=0}^{N}$ of~\eqref{eq:sys}, affected by noise satisfying Assumption~\ref{ass:disturbance_bound}, in order to design $K$ such that the closed-loop matrix $A_K$ is stable and the following quadratic performance specification on~\eqref{eq:sys} is guaranteed for all $A_K\in\Sigma_{X,U}^K$.
\begin{definition}\label{def:quad_perf}
We say that the closed-loop system~\eqref{eq:sys} with state feedback $u_k=Kx_k$ satisfies quadratic performance with index $P=\begin{bmatrix}Q&S\\S^\top&R\end{bmatrix}$, where $R\succeq0$, if there exists an $\varepsilon>0$ such that
\begin{align}\label{eq:quad_perf}
\sum_{k=0}^\infty \begin{bmatrix}w_k\\z_k\end{bmatrix}^\top
\begin{bmatrix}Q&S\\S^\top&R\end{bmatrix}\begin{bmatrix}w_k\\z_k\end{bmatrix}\leq-\varepsilon\sum_{k=0}^\infty w_k^\top w_k
\end{align}
for all $w\in\ell_2$.
\end{definition}
Important special cases of Definition~\ref{def:quad_perf} are $Q=-\gamma^2 I,S=0,R=I$ for the $\mathcal{H}_\infty$-control problem and $Q=0,S=-I,R=0$ for closed-loop strict passivity.
Note that the disturbance $w$ enters the present problem setting in two different ways.
First, it perturbs the measured input-state trajectory during the initial data generation for which $w$ is bounded as $W\in\mathcal{W}$.
Second, it enters the control objective of achieving quadratic performance of the channel $w\mapsto z$.
For instance, a desired $\mathcal{H}_\infty$-performance of this channel corresponds to a robustness objective for the closed loop with respect to noise.

Similar to Section~\ref{sec:state_feedback_stability}, the uncertain closed loop of~\eqref{eq:sys}, including the performance channel $w\mapsto z$, can be written as a lower LFT.
To be more precise, for a state-feedback gain $K=UG$, where $G$ satisfies $I=XG$, a superset of the uncertain closed loop from $w$ to $z$ can be parametrized as
%
\begin{align}\nonumber
\mleft[
\begin{array}{c}x_{k+1}\\\hline z_k\\\tilde{z}_k\end{array}
\mright]&=
\mleft[
\begin{array}{c|cc}X_+G&B_w&B_w\\\hline
C+DUG&D_w&0\\-G&0&0\end{array}
\mright]
\mleft[
\begin{array}{c}x_k\\\hline w_k\\\tilde{w}_k\end{array}
\mright],\\\label{eq:LFT_W_performance}
\tilde{w}_k&=W\tilde{z}_k,
\end{align}
for $W\in\mathcal{W}$.
The above system contains two disturbance inputs: $w$ to model the performance channel $w\mapsto z$, representing the control objective of closed-loop quadratic performance, and  $\tilde{w}$ to model the uncertainty originating from the noisy data, similar to the LFT~\eqref{eq:LFT_W}.
The following result derives state-feedback controllers with robust performance for~\eqref{eq:LFT_W_performance}.

\begin{corollary}\label{cor:performance}
If there exist $\mathcal{X}\succ0,G\in\mathbb{R}^{N\times n},\lambda>0,$ such that~\eqref{eq:cor_performance_LMI} and
\begin{figure*}[b]
\noindent\makebox[\linewidth]{\rule{\textwidth}{0.4pt}}
\vspace{2pt}
\begin{align}\label{eq:cor_performance_LMI}
\mleft[
\begin{array}{ccc}I&0&0\\X_+G&B_w&B_w\\\hline
0&I&0\\C+DUG&D_w&0\\\hline0&0&I\\-G&0&0\end{array}
\mright]^\top
\mleft[
\begin{array}{cc|cc|cc}
-\mathcal{X}&0&0&0&0&0\\
0&\mathcal{X}&0&0&0&0\\\hline
0&0&Q&S&0&0\\
0&0&S^\top&R&0&0\\\hline
0&0&0&0&\lambda Q_w&\lambda S_w\\
0&0&0&0&\lambda S_w^\top&\lambda R_w
\end{array}\mright]
\mleft[
\begin{array}{ccc}I&0&0\\X_+G&B_w&B_w\\\hline
0&I&0\\C+DUG&D_w&0\\\hline0&0&I\\-G&0&0\end{array}
\mright]
\prec0
\end{align}
\end{figure*}
\begin{align}\label{eq:cor_performance_equality}
XG&=I
\end{align} 
hold, then, for any $(A,B)\in\Sigma_{X,U}$,
\begin{itemize}
\item[i)] $A+BK$ with $K=UG$ is stable,
\item[ii)] \eqref{eq:sys} with $u_k=Kx_k$ satisfies quadratic performance with index $P$.
\end{itemize}
\end{corollary}
\begin{proof}
The result follows from known robust control methods (cf.~\cite{Scherer00,Scherer00b}).
\end{proof}

Corollary~\ref{cor:performance} applies robust control methods to design controllers which guarantee robust closed-loop performance for all matrices in $\mathcal{A}_G^s$ and hence, according to Theorem~\ref{thm:uncertain}, for all closed-loop matrices $\Sigma_{X,U}^K$ consistent with the data (compare the discussion after Corollary~\ref{cor:stab}).
This implies that the closed-loop channel $w\mapsto z$ satisfies quadratic performance over an \emph{infinite} time-horizon for \emph{arbitrary} disturbance inputs which are not required to satisfy a bound of the form $W\in\mathcal{W}$, compare Definition~\ref{def:quad_perf}.
In order to achieve this goal, a data trajectory of \emph{finite} length and the (finite-horizon) assumption $\hat{W}\in\mathcal{W}$ on the disturbance generating the data are sufficient.
It is straightforward to extend Corollary~\ref{cor:performance} to design controllers with performance guarantees for the channel $w^p\mapsto z$, where $w^p$ is an exogenous input different from the noise perturbing the initial data trajectory, i.e., $w^p\neq w$.
Further, following the same steps as in Remark~\ref{rk:LMI}, the nonlinear matrix inequality~\eqref{eq:cor_performance_LMI} can be transformed into~\eqref{eq:cor_performance_LMI_Schur}, which is an LMI in the variables $\mathcal{Y}=\mathcal{X}^{-1},M=G\mathcal{X}^{-1}$ for a fixed multiplier $\lambda$.
Thus, the proposed feasibility problem can be solved via a line-search over $\lambda$.
\begin{figure*}[b]
\noindent\makebox[\linewidth]{\rule{\textwidth}{0.4pt}}
\vspace{2pt}
\begin{align}\label{eq:cor_performance_LMI_Schur}
\begin{bmatrix}-\mathcal{Y}&(C\mathcal{Y}+DUM)^\top(RD_w+S^\top)&-\lambda M^\top S_w^\top&M^\top X_+^\top &(C\mathcal{Y}+DUM)^\top & M^\top\\
*&Q+SD_w+D_w^\top S^\top+D_w^\top RD_w&0&B_w^\top&0&0\\
*&*&\lambda Q_w&B_w^\top&0&0\\
*&*&*&-\mathcal{Y}&0&0\\
*&*&*&*&-R^{-1}&0\\
*&*&*&*&*&-(\lambda R_w)^{-1}\end{bmatrix}\prec0
\end{align}
\end{figure*}

\subsection{Systems with partial model knowledge}\label{sec:model}
We conclude the section by presenting an extension of the proposed framework to systems with mixed data-driven and model-based components.
To this end, we consider systems of the form
\begin{align}\label{eq:sys_conn}
\mleft[
\begin{array}{c}x_{k+1}\\\tilde{x}_{k+1}\\\hline
z_k
\end{array}
\mright]=
\mleft[\begin{array}{cc|cc}A_{1}&A_2&B_{w1}&B_{1}\\
A_3&A_4&B_{w2}&B_2\\\hline
C_{1}&C_{2}&D_{w}&D
\end{array}
\mright]
\mleft[
\begin{array}{c}x_{k}\\\tilde{x}_k\\\hline
w_k\\ u_k
\end{array}
\mright],
\end{align}
where the matrices $A_1$ and $B_1$ are unknown, but all other matrices occurring in~\eqref{eq:sys_conn} are known.
Further, a single open-loop data trajectory $\{x_k,\tilde{x}_k,u_k\}_{k=0}^{N-1}$, which is perturbed by some unknown disturbance realization $\hat{W}\in\mathcal{W}$, is available.

In the following, we consider the closed loop of~\eqref{eq:sys_conn} under control with state-feedback $u_k=K_1x_k+K_2\tilde{x}_k$.
Suppose there exist matrices $G_1\in\mathbb{R}^{N\times n},G_2\in\mathbb{R}^{N\times\tilde{n}}$, where $n$ and $\tilde{n}$ are the dimensions of $x_k$ and $\tilde{x}_k$, respectively, such that
\begin{align}\label{eq:combined_equality}
\begin{bmatrix}I&0\\K_1&K_2\end{bmatrix}=
\begin{bmatrix}X\\U\end{bmatrix}\begin{bmatrix}G_1&G_2\end{bmatrix}.
\end{align}
Multiplying~\eqref{eq:combined_equality} from the left by $\begin{bmatrix}A_1&B_1\end{bmatrix}$, we obtain
\begin{align*}
A_1+B_1K_1&=(X_+-A_2\tilde{X}-B_{w1}\hat{W})G_1,\\
B_1K_2&=(X_+-A_2\tilde{X}-B_{w1}\hat{W})G_2.
\end{align*}
These relations allow us to replace all occurrences of the unknown matrices $A_1$ and $B_1$ in the closed-loop dynamics.
Thus, following the same steps as in the previous sections, we obtain the LFT~\eqref{eq:LFT_W_conn} with $W\in\mathcal{W}$, which parametrizes a superset of the uncertain closed loop dynamics of~\eqref{eq:sys_conn} under the above state-feedback.
\begin{figure*}[b]
\noindent\makebox[\linewidth]{\rule{\textwidth}{0.4pt}}
\vspace{2pt}
\begin{align}\label{eq:LFT_W_conn}
\begin{split}
\mleft[
\begin{array}{c}x_{k+1}\\\tilde{x}_{k+1}\\\hline z_k\\\tilde{z}_k\end{array}
\mright]&=
\mleft[
\begin{array}{cc|cc}(X_+-A_2\tilde{X})G_1&A_2+(X_+-A_2\tilde{X})G_2&B_{w1}&B_{w1}\\
A_3+B_2UG_1&A_4+B_2UG_2&B_{w2}&0\\\hline
C_1+DUG_1&C_2+DUG_2&D_w&0\\-G_1&-G_2&0&0\end{array}
\mright]
\mleft[
\begin{array}{c}x_k\\\tilde{x}_k\\\hline w_k\\\tilde{w}_k\end{array}
\mright]\\
\tilde{w}_k&=W\tilde{z}_k
\end{split}
\end{align}
\end{figure*}
Note that this LFT depends only on known matrices and the open-loop data trajectory $\{x_k,\tilde{x}_k,u_k\}_{k=0}^{N-1}$.
The structure of~\eqref{eq:LFT_W_conn} resembles that of the LFT~\eqref{eq:LFT_W_performance} and therefore, robust controllers for the mixed system~\eqref{eq:sys_conn} can be derived by proceeding as in Section~\ref{sec:state_feedback_performance}.

In contrast to the previous sections, the condition~\eqref{eq:combined_equality} requires not only that $X$ has full row rank but also that $N\geq n+\tilde{n}$.
Moreover, if $\begin{bmatrix}X^\top &U^\top\end{bmatrix}^\top$ has full row rank, i.e., the data-driven component of~\eqref{eq:sys_conn} is persistently exciting, and $N\geq n+\tilde{n}$, then, for any matrices $K_1$ and $K_2$, there exist matrices $G_1$ and $G_2$ satisfying~\eqref{eq:combined_equality}, i.e., any controller can be constructed.

\begin{remark}
Our original motivation for considering the above mixed data-driven and model-based configuration comes from $\mathcal{H}_\infty$-loop-shaping:
The $\mathcal{H}_\infty$-control problem is usually not solved for the performance channel $w\mapsto z$ directly, but rather for the channel $w\mapsto z^f$, where $z^f$ is the output of a filter with input $z$.
In this scenario, the known components of~\eqref{eq:sys_conn} are mainly that of the filter, whereas the unknown matrices $(A_1,B_1)$ are equal to $(A_{tr},B_{tr})$ from~\eqref{eq:sys}.
Notably, $A_2=0$ holds in this case and hence, it can be seen from~\eqref{eq:LFT_W_conn} that measured data of the filter state $\tilde{x}$ is not required.
To conclude, by iteratively refining the filter dynamics and solving the robust performance design problem for the LFT~\eqref{eq:LFT_W_conn}, we can systematically perform loop-shaping for the system~\eqref{eq:sys}, without knowledge of $(A_{tr},B_{tr})$.
\end{remark}

\section{Example}\label{sec:example}
In this section, we apply the results of Section~\ref{sec:state_feedback} to the robust $\mathcal{H}_\infty$-control problem for an unstable example system.
We consider System~\eqref{eq:sys} with
\begin{align*}
A_{tr}&=\begin{bmatrix}-0.5&1.4&0.4\\
-0.9&0.3&-1.5\\
1.1&1&-0.4\end{bmatrix},\>\>B_{tr}=\begin{bmatrix}0.1&-0.3\\-0.1&-0.7\\0.7&-1
\end{bmatrix},\\
B_w&=I_3,\>\>C=I_3,\>\>D=0,\>\>D_{w}=0,
\end{align*}
where it is assumed that $A_{tr}$ and $B_{tr}$ are not available.
We generate data $\{x_k,u_k\}_{k=0}^{N}$ of length $N=20$ by sampling the input $u_k$ uniformly from $[-1,1]^2$ and the disturbance $\hat{w}$ uniformly from the ball $\lVert \hat{w}\rVert_2\leq\bar{w}$, where $\bar{w}=0.02$.
This implies the disturbance bound $\hat{W}\in\mathcal{W}$ for $Q_w=-I,S_w=0,R_w=\bar{w}^2I$.
In the following, we compute a state-feedback gain via Corollary~\ref{cor:performance} to achieve robust closed-loop quadratic performance with index $P=\begin{bmatrix}-\gamma^2I&0\\0&I\end{bmatrix}$ for a possibly small $\gamma>0$, i.e., a small $\mathcal{H}_\infty$-norm of $w\mapsto z$.
Following the procedure described in Remark~\ref{rk:LMI}, we verify that~\eqref{eq:cor_performance_LMI} and~\eqref{eq:cor_performance_equality} are feasible for $\gamma=2.4$ and we obtain a corresponding controller as $K=\begin{bmatrix}-2.45&-1.29&-2.4\\-0.61&-0.03&-2.18\end{bmatrix}$, which leads to a closed-loop $\mathcal{H}_\infty$-norm of $2.3$.
In contrast, the minimal achievable $\mathcal{H}_\infty$-norm using a nominal (model-based) state-feedback is $2.2$.
Thus, the proposed approach yields a controller with guaranteed performance close to the ideal case with full model knowledge, despite noisy measurements.
For larger noise levels $\bar{w}\geq0.04$, the design problem is infeasible since it addresses performance guarantees \emph{for all matrices} in the set $\Sigma_{A,B}^K$, which grows with $\bar{w}$.

In the following, we analyze the influence of the data length $N$ on the feasibility of~\eqref{eq:cor_performance_LMI} and~\eqref{eq:cor_performance_equality} for the above design problem.
To keep the signal-to-noise ratio (approximately) constant, we modify the bound $\bar{w}$ of the noise generating the data linearly with $N$, i.e., $\bar{w}=\frac{0.02}{20}N$.
For each data horizon $4\leq N\leq 20$, we perform $100$ experiments to generate data for the controller design, each with different (random) inputs $u$ and disturbances $\hat{w}$ as described above.
Figure~\ref{fig:loop_pe} shows the number of successful designs depending on $N$.
It can be observed that the feasibility of~\eqref{eq:cor_performance_LMI} and~\eqref{eq:cor_performance_equality} is enhanced if $N$ increases, and $N\geq15$ suffices to successfully design a controller from $100$ out of $100$ experiments.
Intuitively, this can be explained by noting that, with an increasing number of data points, the size of $\Sigma_{X,U}$ decreases, and Corollary~\ref{cor:performance} provides robust performance guarantees for (a superset of) the uncertain closed loop matrices $\Sigma_{X,U}^K$ consistent with the data.
Moreover, even for $N$ as low as $4$, in which case the data are not persistently exciting, the design is successful in more than $50\%$ of the scenarios.

\begin{figure}[H]
\begin{center}
\includegraphics[width=0.516\textwidth]{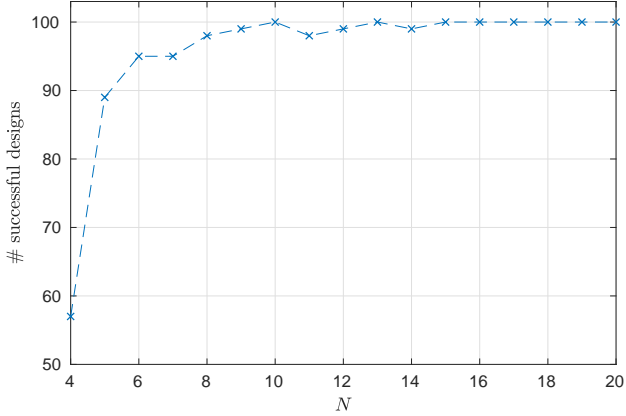}
\end{center}
\caption{Number of successful designs for which~\eqref{eq:cor_performance_LMI} and~\eqref{eq:cor_performance_equality} are feasible for the present example, depending on the data length $N$.
For each horizon $N$, $100$ experiments are carried out with varying random inputs and disturbances to generate data for controller design according to Corollary~\ref{cor:performance}.}
\label{fig:loop_pe}
\end{figure}

Finally, we comment on the computational complexity of the feasibility problem stated in Corollary~\ref{cor:performance}.
After its reformulation (cf. Remark~\ref{rk:LMI}), the problem contains an LMI with $2(n+m_w)+p_z+N$ rows, i.e., it is of size $35\times 35$ for the above example with $N=20$, as well as an equality constraint of size $n\times n=9$.
Moreover, the matrix variables\footnote{Note that $\mathcal{Y}$ is symmetric and has therefore only $\frac{n(n+1)}{2}$ \emph{free} decision variables.} $\mathcal{Y}$ and $M$ are of size $n\times n=3\times 3$ and $N\times n=20\times 3$, respectively.
The complexity of standard LMI solvers scales cubically with the number of decision variables.
Thus, the proposed controller design method scales cubically with the data length $N$ and proportionally to $n^6$ if $n$ is the system dimension, similar as in model-based robust controller design.

\section{Conclusion}\label{sec:conclusion}
The present paper provides direct, data-driven design procedures for state-feedback gains, which achieve guaranteed closed-loop stability and performance, using noisy input-state data.
Based on a data-driven parametrization of the closed-loop matrices that are consistent with the data, known robust control methods can be applied.
The parametrization is extended to a setting with partial model knowledge, and the design procedures are applied successfully to an unstable example system.
The proposed approach leads to \emph{end-to-end guarantees} for the closed loop, using a single noisy open-loop data trajectory of finite length, and is thus a promising alternative to sequential system identification and robust control.
Future research should extend the results of this paper to robust data-driven output-feedback control.

\addtolength{\textheight}{-12cm}   

\bibliographystyle{IEEEtran}
\bibliography{Literature}

\begin{thebibliography}{10}
\providecommand{\url}[1]{#1}
\csname url@rmstyle\endcsname
\providecommand{\newblock}{\relax}
\providecommand{\bibinfo}[2]{#2}
\providecommand\BIBentrySTDinterwordspacing{\spaceskip=0pt\relax}
\providecommand\BIBentryALTinterwordstretchfactor{4}
\providecommand\BIBentryALTinterwordspacing{\spaceskip=\fontdimen2\font plus
\BIBentryALTinterwordstretchfactor\fontdimen3\font minus
  \fontdimen4\font\relax}
\providecommand\BIBforeignlanguage[2]{{%
\expandafter\ifx\csname l@#1\endcsname\relax
\typeout{** WARNING: IEEEtran.bst: No hyphenation pattern has been}%
\typeout{** loaded for the language `#1'. Using the pattern for}%
\typeout{** the default language instead.}%
\else
\language=\csname l@#1\endcsname
\fi
#2}}

\bibitem{Hou13}
Z.-S. Hou and Z.~Wang, ``From model-based control to data-driven control:
  Survey, classification and perspective,'' \emph{Information Sciences}, vol.
  235, pp. 3--35, 2013, 

\bibitem{Recht18}
B.~Recht, ``A tour of reinforcement learning: The view from continuous
  control,'' \emph{Annual Review of Control, Robotics, and Autonomous Systems},
  2018.

\bibitem{Matni19b}
N.~Matni and S.~Tu, ``A tutorial on concentration bounds for system
  identification,'' \emph{arXiv preprint arXiv:1906.11395}, 2019.

\bibitem{Matni19}
N.~Matni, A.~Proutiere, A.~Rantzer, and S.~Tu, ``From self-tuning regulators to
  reinforcement learning and back again,'' \emph{arXiv preprint
  arXiv:1906.11392}, 2019.

\bibitem{Boczar18}
R.~Boczar, N.~Matni, and B.~Recht, ``Finite-data performance guarantees for the
  output-feedback control of an unknown system,'' in \emph{Proc. 57th IEEE
  Conf. on Decision and Control}, 2018, pp. 2994--2999.

\bibitem{Dean19}
S.~Dean, H.~Mania, N.~Matni, B.~Recht, and S.~Tu, ``On the sample complexity of
  the linear quadratic regulator,'' \emph{Foundations of Computational
  Mathematics}, 2019, https://doi.org/10.1007/s10208-019-09426-y.

\bibitem{Milanese91}
M.~Milanese and A.~Vicino, ``Optimal estimation theory for dynamic systems with
  set membership uncertainty: an overview,'' \emph{Automatica}, vol.~27, no.~6,
  pp. 997--1009, 1991.

\bibitem{Kosut01}
R.~L. Kosut, ``Uncertainty model unfalsification for robust adaptive control,''
  \emph{Annual Reviews in Control}, vol.~25, pp. 65--76, 2001.

\bibitem{Willems05}
J.~C. Willems, P.~Rapisarda, I.~Markovsky, and B.~{De Moor}, ``A note on
  persistency of excitation,'' \emph{Systems \& Control Letters}, vol.~54, pp.
  325--329, 2005.

\bibitem{Romer19}
A.~Romer, J.~Berberich, J.~K{\"o}hler, and F.~Allg{\"o}wer, ``One-shot
  verification of dissipativity properties from input-output data,'' \emph{IEEE
  Control Systems Letters}, vol.~3, no.~3, pp. 709--714, 2019.

\bibitem{Berberich20a}
J.~Berberich and F.~Allg\"ower, ``A trajectory-based framework for data-driven
  system analysis and control,'' in \emph{Proc. European Control Conference},
  2020, to appear, preprint online: arXiv:1903.10723.

\bibitem{Persis19}
C.~{De Persis} and P.~Tesi, ``Formulas for data-driven control: Stabilization,
  optimality and robustness,'' \emph{arXiv:1903.06842}, 2019.

\bibitem{Waarde19}
H.~J. {van Waarde}, J.~Eising, H.~L. Trentelman, and M.~K. Camlibel, ``Data
  informativity: a new perspective on data-driven system analysis and
  control,'' \emph{arXiv:1908.00468}, 2019.

\bibitem{Zhou96}
K.~Zhou, J.~C. Doyle, and K.~Glover, \emph{Robust and optimal control}.\hskip
  1em plus 0.5em minus 0.4em\relax Prentice-Hall, Inc., Englewood Cliffs, N.J.,
  1996.

\bibitem{Scherer00}
C.~Scherer and S.~Weiland, \emph{Linear Matrix Inequalities in Control},
  3rd~ed.\hskip 1em plus 0.5em minus 0.4em\relax New York: Springer-Verlag,
  2000.

\bibitem{Scherer00b}
C.~Scherer, ``Robust mixed control and linear parameter-varying control with
  full-block scalings,'' in \emph{Advances in Linear Matrix Inequality Methods
  in Control}.\hskip 1em plus 0.5em minus 0.4em\relax SIAM: Philadelphia, 2000,
  pp. 187--207.

\bibitem{Safonov97}
M.~G. Safonov and T.-C. Taso, ``The unfalsified control concept and learning,''
  \emph{IEEE Transactions on Automatic Control}, vol.~42, no.~6, pp. 843--847,
  1997.

\end{thebibliography}
\end{document}